\documentclass[final,12pt,times]{elsarticle}

\usepackage{amssymb}   
\usepackage{amsthm}    
\usepackage{amsmath}
\usepackage{xcolor}
\usepackage{enumerate} 
\usepackage{graphicx}
\usepackage[english]{babel}
\usepackage{hyphenat}
\usepackage[shortcuts]{extdash} 
\usepackage{kantlipsum}
\usepackage{paralist}
\biboptions{authoryear}
\allowdisplaybreaks

\def\K{\mathcal{K}}

\def\KL{\mathcal{KL}}
\def\R{\mathbb{R}}
\def\N{\mathbb{N}}
\def\X{\mathcal{X}}

\def\T{\mathcal{T}}
\def\mer{\hfill $\circ$}

\def\qed{$\hfill\blacksquare$}
\def\es{\notag \\}
\def\lp{\left(}
\def\rp{\right)}
\def\lb{\left|}
\def\rb{\right|}

\newtheorem{theorem}{Theorem}[section]
\newtheorem{definition}[theorem]{Definition}

\newtheorem{lemma}[theorem]{Lemma}

\newtheorem{assumption}[theorem]{Assumption}

\newtheorem{prop}[theorem]{Proposition}
\theoremstyle{remark}

\journal{} 

\begin{document}

\begin{frontmatter}

\title{{\color{black}On} the stability of nonlinear sampled-data systems and their continuous-time limits}


\address{Centro Internacional Franco-Argentino de Ciencias de la Informaci\'on y de Sistemas (CIFASIS),
CONICET-UNR, Ocampo y Esmeralda, 
2000 Rosario, Argentina. 
{\texttt{\{vallarella,haimovich\}@cifasis-conicet.gov.ar}}}

\begin{abstract}
{\color{black}
This work deals with the stability analysis of nonlinear sampled-data systems under nonuniform
sampling. It establishes novel relationships between the stability property of the exact discrete-time
model for a given sequence of (aperiodic) sampling instants and the stability property of the continuous-time system when the maximum admissible sampling period converges to zero. These results can be used to infer stability properties for the sampled-data system by direct inspection of the stability of
the mentioned continuous-time system, a task which is typically easier than the analysis of the closed-loop sampled-data system. Compared to the literature, our results allow to prove stronger (asymptotic) sampled-data stability properties for nonlinear systems in cases for which existing results only guarantee practical stability.

}
\end{abstract}

\end{frontmatter}

\section{Introduction}
\label{sec:introduction}

{\color{black}
The two main
approaches to design controllers for sampled-data nonlinear systems are:
a) to design a discrete-time (DT) controller based
on a DT model of the plant 
\citep{NesicSCL99,NesicTAC04,liu2008input,nevsic2009stability,ustunturk2012output,ustunturk2013backstepping,noroozi2018integral,beikzadeh2018robust,IEEETAC18}
or b) to obtain the controller by adequate discretization of a continuous-time (CT) one. 
In the first approach, the DT model is usually an approximation
due to the impossibility of solving nonlinear differential equations in closed form.
In the second approach, emulation of the CT controller 
is known to stabilize  the sampled\-/data system
under sufficiently fast sampling 
\citep{nesic2009explicit,proskurnikov2020does}. 
In addition, controller redesign 
that accounts for sampling may lead to better performance 
\citep{NESIC20051143,monaco2007advanced,GRUNE2008225,POSTOYAN20082099}.
In both approaches, the final result is in general a sampling period\-/dependent DT control law to be
implemented 
usually under zero-order hold.

Results along the first approach generally
give conditions under which a stability property 
of the approximate DT model in closed loop is enough to guarantee 
(some type of) stability of the closed-loop sampled-data system 
under sufficiently fast sampling.
In many cases, only practical (not asymptotic) stability 
is ensured or else strong conditions are imposed \cite[see][and references therein]{valcar_auto20}. %
Uniform (periodic) sampling is usually considered, but some results ensuring only practical stability also allow nonuniform (aperiodic) sampling \citep{di2019sampled,di2021exponential}. 
Works that address asymptotic stability of sampled-data systems and do not fall within any of the two approaches mentioned also exist \citep{li2018output,lin2018semiglobal,8928956,lin2021new}. \citet{li2018output} shows how asymptotic stability can be preserved via the selection of a specific sampling period. \citet{lin2018semiglobal,8928956,lin2021new} give results for
semiglobal asymptotic stabilizability under constant sampling,
where the term ``semiglobal'' involves possibly different convergence rates for different sets of initial conditions and upper bounds on the sampling period. By contrast, previous own works within the first approach \citep{IEEETAC18,vallarella2018characterization,valcar_auto20} 
address semiglobal stability properties, where ``semiglobal'' involves the same convergence rate but possibly different maximum sampling periods for every initial condition, and 
allow nonuniform sampling \citep{omran2016stability,HETEL2017309}. 
In particular, we derived conditions under which semiglobal (practical or exponential) stability
is carried over between different DT models \citep{IEEETAC18,valcar_auto20}
and we provided adequate Lyapunov-type guarantees \citep{vallarella2018characterization}.
This allows to establish the stability of the exact DT model,
i.e. the model whose state coincides with 
the state of the sampled\-/data system at sampling instants and which is not assumed to be
available.

In this context, the aim of this note is to
establish a precise correspondence
between
asymptotic (not only practical) stability properties under nonuniform sampling of the 
exact DT model 
and stability properties of a related CT system.
This CT system is just the CT open-loop system in closed-loop
with the CT limit of the control law, the latter being the value of the control action when the sampling period tends to zero.
We provide mild conditions
on the open-loop models and (sampling period\-/dependent) control laws
under which the closed-loop exact DT model exhibits some asymptotic stability property if and only if 
the related CT system exhibits the respective CT equivalent property.
These conditions complement existing results by allowing to guarantee stronger (asymptotic) stability 
properties for the sampled\-/data system in cases where previous results 
only show practical stability or apply to uniform sampling. 
The mild conditions required admit systems that are not globally Lipschitz and not necessarily input-affine, and control laws that may be not differentiable with respect to the sampling period. 

\textbf{Notation}:
$\R$, $\R_{\geq 0}$, $\N$ 
and $\N_0$ denote the reals, nonnegative reals, naturals and nonnegative integers.
Classes $\K$, $\K_\infty$ and $\KL$ of functions are defined as in \cite{khalil2002nonlinear}.
For a vector $x \in \R^n$, $|x|$ denotes its Euclidean norm.
A sequence is noted as $\{T_i\}:=\{T_i\}_{i=0}^{\infty}$. 
For any $\{T_i\} \subset \R_{\ge 0}$,
we define $\sum_{i=0}^{-1} T_i = 0$.
Given $T>0$, we define
$\Phi(T):=\{ \{T_i\} : \{T_i\} \text{ is such that }  T_i \in (0,T) \text{ for all } i\in \N_0 \}$.

}

\section{Problem statement}
\label{sec:preliminaries}
{\color{black}
Our aim is to obtain mild conditions
that preserve stability properties for sampled-data systems that
arise from nonlinear plants of the form
\begin{equation}
    \dot x =  f(x,u), 
    \label{12342} 
\end{equation} under zero-order hold,
where $x(t) \in \R^n$, $u(t) \in \R^m$ are the state and control vectors. 
A control law $u_c :\R^n \rightarrow \R^m$ may render  
\begin{gather}
  \dot x = f(x,u_c(x))=:h(x) 
  \label{lazocer2}
\end{gather}
stable as per one of the following definitions.
}
\begin{definition}
\label{def:conti}
The system \eqref{lazocer2} is said to be
\begin{enumerate}[i)]
\itemsep0em
\item \label{item:gas} Globally Asymptotically Stable (GAS) 
if there exists $\beta \in \KL$ such that for any $x_0 \in \R^n$ the solutions satisfy
$|x(t)|\leq  \beta(|x_0|,t), \medspace \forall t\geq 0$.
If additionally $\beta \in \KL$ can be chosen as 
$\beta(r,t) := K r e^{-\lambda t}$ with $K \geq 1$ and $\lambda > 0$ it is said to be Globally Exponentially Stable (GES).
\item \label{item:les} Locally Exponentially Stable (LES) 
if there exist $K\geq 1$ and $R,\lambda>0$ such that for all $|x_0|\leq R$ the solutions satisfy
$|x(t)|\leq K |x_0|e^{-\lambda t}, \medspace \forall t\geq 0$.
\item GALES if it is GAS and LES. 
\end{enumerate}
\end{definition}

We consider that the function $f$ in \eqref{12342}
and a control law $u_c$ in \eqref{lazocer2} 
fulfill the following local Lipschitzness assumptions.

\begin{assumption}
\label{ass:B1}
$f:\R^n \times \R^m \rightarrow \R^n$ 
fulfills $f(0,0)=0$ and 
for every $M,M_u\geq0$
there exists
$L=L(M,M_u)>0$ 
such that for all 
$|x|,|y|\leq M$ and $|u|,|v|\leq M_u$
we have 
$\left|f(x,u)-f(y,v)\right| \leq L(|x-y|+|u-v|)$. 
\end{assumption}

\begin{assumption}
\label{ass:B2} 
$u_c:\R^n \rightarrow \R^m$ 
fulfills $u_c(0)=0$ and 
for every $M \geq 0$
there exists
$L=L(M)>0$ such that for all $|x|,|y| \leq M$
we have 
$\left|u_c(x)-u_c(y)\right| \leq L|x-y|$.
\end{assumption}

{\color{black}
We consider sampling instants 
$t_k$, $k\in \N_0$, $t_0 = 0$ and $t_{k+1} = t_k + T_k$, 
where $T_k>0$ is the $k^{th}$
sampling period. 
The sampling periods may vary following any possible sequence as long as they are bounded
by a maximum admissible sampling interval;
we refer to this situation as Varying Sampling Rate (VSR). 
We assume that $T_k$ 
is either known or 
determined at instant $t_k$,
so that this information may be used to perform 
the current control action: 
$u_k=U(x_k,T_k)$.
This is always the case under periodic sampling (i.e. $T_k \equiv T > 0$), where the control law is designed based on prior knowledge of the sampling period.
The sampled\-/data system that arises from \eqref{12342}
in feedback with 
$u_k=U(x_k,T_k)$
under zero-order hold
is
\begin{equation}
 \dot x(t)=  f(x(t),U(x(t_k),T_k))  \text{ } \medspace \forall t \in [t_k,t_{k+1}), \medspace k\in \N_0. \label{SDS}
\end{equation}
We consider DT models of \eqref{SDS}.
These can be regarded as estimates of the value $x_{k+1} = x(t_{k+1})$,
given $x_k$ and $u_k$ at the sampling instant $t_k$,
namely $x_{k+1}= F(x_k,u_k,T_k)$. 
The \textit{exact} DT model is the one that generates the actual value that $x(t_{k+1})$ will have as the solution of \eqref{SDS}, and is denoted by $F^e$. For nonlinear plants, the exact DT model
may be unavailable due to the 
difficulty or impossibility
of solving nonlinear differential equations. 
Thus, a suitable design approach is to design the control law based on a sufficiently good \textit{approximate} DT model
of the plant, such as Runge-Kutta models \citep{StuartDynSysAndNumAn}. 
The simplest of these models, the Euler model, is given by
$F^{E}(x,u,T):=x+T f(x,u)$. 
%
For a DT model $F$ and control law $U$, we define the closed-loop DT model
$\bar F_U(x,T):=F(x,U(x,T),T)$. We may also simply denote a closed-loop DT model by $\bar F$ when the control law is not
important in the context.

To state our results we need the following 
Equilibrium\-/Preserving Consistency (EPC) property, which 
bounds the mismatch 
between any two of the previously defined DT models' solutions 
after one sampling interval.
This property
becomes equivalent to the REPC property in \cite{valcar_auto20} when no errors affect the control input. 
}

\begin{definition}
\label{def:EPC}
The DT model $\bar F^a$ 
is said to be  Equilibrium\-/Preserving Consistent (EPC)
with  $\bar F^b$ if
for each $M\geq 0$
there exist constants 
{\color{black}$K:=K(M)>0$, $T^*:=T^*(M)>0$} and a function $\rho \in \K_\infty$
such that 
\begin{equation}
\displaystyle
\left| \bar F^a (x,T)- \bar F^b(y,T) \right|    
\leq (1+KT)\left|x-y\right|+ T \rho(T)  \max\{|x|,|y|\} \label{eq:EPC}
\end{equation}
for all $|x|,|y| \leq M$ and $T\in (0,T^*)$.
The pair $(\bar F^a, \bar F^b)$ is said to be EPC if $\bar F^a$ is EPC with $\bar F^b$.
\end{definition}
{\color{black}
The EPC property is sufficient to ensure
that the following stability properties for DT models, suitable
under nonuniform sampling \citep[see][]{IEEETAC18,vallarella2018characterization,valcar_auto20},
are shared between different DT models. This fact is stated in 
Theorem~\ref{theorem:SEISS1}.
}
\begin{definition}
\label{def:sampleddef}
The system $x_{k+1}=\bar F(x_k,T_k)$ is said to be
\begin{enumerate}[i)]
\itemsep0em

\item Semiglobally Practically Stable-VSR (SPS-VSR) 
if there exists $\beta \in \KL$
such that for every $M\geq0 $ and $R>0$ 
there exists $T^\star:=T^\star(M,R)>0$
such that for all 
$k\in \N_0$, $\{T_i\}\in \Phi(T^\star)$ and $|x_0|\leq M$ the solutions satisfy
$|x_k|\leq \beta\lp|x_0|,\sum_{i=0}^{k-1}T_i\rp+R$.

\item Locally Exponentially Stable-VSR (LES-VSR) 
if there exist $K\geq 1$ and $R,T^\star,\lambda>0$ such that for all
$k\in \N_0$, $\{T_i\}\in \Phi(T^\star)$ and $|x_0|\leq R$ the solutions satisfy $|x_k|\leq K |x_0|e^{-\lambda \sum_{i=0}^{k-1}T_i}$.

\item Semiglobally (asymptotically) and Locally Exponentially Stable (SLES-VSR) if it is SPS-VSR and LES-VSR. \label{SLES-VSR}

\item Semiglobally (asymptotically) Stable-VSR (SS-VSR) 
if there exists $\beta \in \KL$
such that for every $M\geq0$ 
there exists $T^\star:=T^\star(M)>0$
such that for all 
$k\in \N_0$, $\{T_i\}\in \Phi(T^\star)$ and $|x_0|\leq M$ the solutions satisfy
   $|x_k|\leq \beta\lp|x_0|,\sum_{i=0}^{k-1}T_i\rp$.
If additionally $\beta \in \KL$ can be chosen as $\beta(r,t):=K r \exp(-\lambda t)$ with $K\geq 1$ and $\lambda>0$ it is said to be \textit{Semiglobally Exponentially Stable-VSR} (SES-VSR).
\end{enumerate}
\end{definition}

\begin{theorem}
\label{theorem:SEISS1}
Suppose that $(\bar F^a,\bar F^b)$ is EPC. Then
\begin{inparaenum}[i)]
\itemsep0em
 \item $\bar F^a$ is SPS-VSR $\Leftrightarrow$ $\bar F^b$ is SPS-VSR. \label{item:SPS12}
 \item $\bar F^a$ is LES-VSR $\Leftrightarrow$ $\bar F^b$ is LES-VSR.\label{item:LES12}
 \item $\bar F^a$ is SLES-VSR $\Leftrightarrow$ $\bar F^b$ is SLES-VSR. \label{item:SLES-VSR}
 \end{inparaenum}
\end{theorem}

{\color{black}
The proof of Theorem \ref{theorem:SEISS1} follows from the proofs of \citet[Theorem~3.5, Lemma 3.7]{valcar_auto20} and \citet[Theorem 3.1]{IEEETAC18}, imposing that no errors affect the control input
and can be found in Section~\ref{sec:proofs}. 
%
To establish that an approximate model $F^a$ in closed loop with a control law $U(x,T)$ is such that $\bar F^a_U$ is EPC with the exact model $\bar F^e_U$, one could prove EPC of $(\bar F^{E}_U, \bar F^a_U)$, with $\bar F^E$ the Euler model, which is much simpler. This is all that is required because EPC is a transitive property and $(\bar F^e_U,\bar F^{E}_U)$
is already known to be EPC \citep{valcar_auto20}. Once $(\bar F^a_U, \bar F^e_U)$ EPC is established, application of Theorem~\ref{theorem:SEISS1} would ensure that a stability property of the approximate model also holds for the exact model.
}

\section{Main results}
\label{sec:mainresults}

\sloppy {\color{black} In this section,
we present mild sufficient conditions
under which sampled-data stability properties as per Definition \ref{def:sampleddef} hold if and only if
the closed-loop CT system \eqref{lazocer2}, with $u_c$ equal to a specific limit of the DT control law, 
has a corresponding stability property.
Conditions to ensure that different open-loop models (say $F^a$ and $F^b$) in closed\-/loop with the same control law (say $U$), are EPC, i.e. $(\bar F^a_U,\bar F^b_U)$ is EPC, already exist \citep{valcar_auto20}. 
To derive our main results we need an extension to the case where different control laws may be used, for which 
we require the following consistency and regularity conditions.}
\begin{definition}
\label{d2}
The pair $(U,V)$ is said to be Semiglobally small-time convergent Consistent (StC) if
for each $M\geq0$ there exist a function $\rho \in \K_\infty$ 
and
$T^*:=T^*(M)>0$ such that
for all $|x|\leq M$ and $T\in[0,T^*)$
we have
\begin{equation}
|U(x,T)-V(x,T) |\leq \rho(T) |x|.  \label{eq:StC}
\end{equation}
\end{definition}

\begin{definition}
\label{d1}
The function $U$ is said to be
Semiglobally small-time Lipschitz (StL)
if 
for each $M\geq0$ there exist $K:=K(M)>0$, $T^*:=T^*(M)>0$
with $T^*(\cdot)$ nonincreasing
such that 
for all $|x|,|y|\leq M$ 
and $T\in[0,T^*)$ we have $U(0,T)=0$ and
\begin{equation}
|U(x,T)-U(y,T) |\leq K|x-y|. \label{eq:StL}  
\end{equation}
\end{definition}

\begin{definition}
\label{def:24}
The DT model $F^a$ is said to be 
Semiglobally small-time Lipschitz Consistent (StLC)
if
for each $M,E\geq0$  there exist $K:=K(M,E)>0$, $T^*:=T^*(M,E)>0$ such that 
for all  $|x|,|y|\leq M$, $|u|,|v|\leq E$, and $T\in[0,T^*)$
\begin{align}
| F^a(x,u,T)-  F^a(y,v,T)|\leq (1+KT) |x-y|+ KT |u-v|.  \label{eq:StLC}
\end{align}
\end{definition}

{\color{black}
Theorem \ref{suficientes1} 
gives sufficient conditions 
so that closed-loop models arising from feeding back the same open-loop model with different control laws are EPC. 
Theorems \ref{theorem:SEISS1} and \ref{suficientes1} are both needed to prove
Theorem \ref{TH6}, which is our main result.
The proofs of Theorems~\ref{suficientes1} and~\ref{TH6} are provided in Section~\ref{sec:proofs}.
}
\begin{theorem}[Proof in Section \ref{app:Transitividad}]
\label{suficientes1}
Suppose that 
\begin{inparaenum}[i)]
\itemsep0em
\item $F^a$ is StLC,
\item $U$ is StL,
\item $(U,V)$ is StC.
\end{inparaenum}
Then $(\bar F^a_U,\bar F^a_V)$ is EPC.
\end{theorem}

\begin{theorem}[Proof in Section \ref{app:theorem7}]
\label{TH6}
Consider system \eqref{12342}, a DT control law $U(x,T)$, and the CT control law $u_c(x) := U(x,0)$.
Suppose that
\begin{enumerate}[i)]
\itemsep0em
    \item Assumptions \ref{ass:B1} and \ref{ass:B2} hold.
    \item $U$ is StL.
   \item $(U,U_c)$ is StC, with $U_c(x,T):=u_c(x)$ for all $T\geq 0$.
\end{enumerate}
Then the CT closed\-/loop plant \eqref{lazocer2} is
\begin{enumerate}[a)]
\itemsep0em
    \item \phantom{GA}LES $\Leftrightarrow$ the exact model $\bar F^e_{U}$ is LES-VSR.\label{item:2_H}
    \item GALES $\Leftrightarrow$ the exact model $\bar F^e_{U}$ is SLES-VSR. \label{item:123}
    \item \phantom{AL}GES $\Leftrightarrow$ the exact model $\bar F^e_{U}$ is SES-VSR.  \label{item:1_H}
\end{enumerate}
\end{theorem}

{\color{black}
Previous results ensure that 
semiglobal practical \citep{NesicSCL99,IEEETAC18}
or semiglobal exponential \citep{valcar_auto20} stability
exhibited by an approximate open\-/loop model $F^a$
in closed\-/loop with a control law $U(x,T)$ is
carried over to the exact closed\-/loop model $\bar F^e_U$.
By Theorem~\ref{TH6} it may be possible to establish even
stronger stability properties of the same exact model $\bar F^e_U$ by analyzing the stability of the CT system \eqref{lazocer2}
in closed\-/loop with the CT limit of $U$, i.e. $u_c(x):=U(x,0)$.
For example, it is known that \eqref{lazocer2} is LES if and only if its linearization 
at the equilibrium is GES. If $\bar F^e_U$ is already known to be SPS-VSR,
and \eqref{lazocer2} is LES, then according to Theorem~\ref{TH6} 
$\bar F^e_U$ is LES-VSR and hence indeed SLES-VSR.
This was not covered in previous results and can be done
without an explicit expression of the exact model.
We emphasize (Theorem~\ref{theorem:SEISS1}) that $(\bar F^a, \bar F^e)$ being EPC does not guarantee that if $\bar F^a$ is SS-VSR then $\bar F^e$ also is, unless stability is locally exponential in addition to asymptotic. 
As a side comment, note that the sampled-data control law $U_c(x,T) = u_c(x) = U(x,0)$ is just the emulation of the CT law $u_c(x)$, i.e. application of the control action $u_c(x)$ irrespective of the sampling period.

Additionally, Theorem~\ref{TH6}
allows to prove stability properties of the sampled-data setting for a broader family of control laws
than the ones used in some controller redesign approaches \citep{monaco2007advanced,GRUNE2008225}.
In particular, note that the StC property in Definition~\ref{d2} allows $U$ to lack a
polynomial expansion of the form $U(x,T)=\sum_{i=0}^N T^i u_i(x)$ with $u_c(x):=u_0(x)$.
For example, consider
$U(x,T):=U_c(x,T)-\sqrt{T}x$ which is not differentiable at $T=0$ but for which
$(U,U_c)$ is StC with $\rho(T):=\sqrt{T}$. 
}

\section{Example}
\label{sec:example}

In the present example, 
we use the provided results to prove the SLES-VSR of a nonlinear sampled\-/data system
in closed\-/loop with a proposed sampling period-dependent DT control law.
Consider the following version of the nonlinear CT plant 
of the form $\dot x= f(x,u)= \tilde f(x)+g(x)u$  
presented  in \cite{nesic2009explicit}
\begin{align}
\begin{bmatrix}
\dot x_1 \\
\dot x_2 
\end{bmatrix}=
\begin{bmatrix}
-2-x_1^2 & 1 \\
0 & x_2(1-x_2) 
\end{bmatrix}  
\begin{bmatrix}
 x_1 \\
 x_2 
\end{bmatrix}+
\begin{bmatrix}
 0 \\
 1 
\end{bmatrix}u. \label{continuo}
\end{align}
Consider 
the controller $u_c(x)=-2x_2$, where $x:=[x_1, x_2]^T$.
First, we will prove that the resulting closed\-/loop plant 
\begin{equation}
\dot x = 
\begin{bmatrix}
-2x_1-x_1^3+x_2 \\
-2x_2+x_2^2(1-x_2)
\end{bmatrix}=\tilde f(x)+g(x)u_c(x):= h(x) \label{este}
\end{equation}
is GALES.
Define the Lyapunov function 
$V(x):={\frac {1}{2}} (x_1^2+x_2^2)$
and note that
\begin{align}
    \dot V(x)&= \frac{\partial V}{\partial x} h(x) 
    =-x_1^4-1.5x_1^2-\frac{1}{2}(x_1-x_2)^2-1.5x^2_2+x_2^3-x_2^4 \es
    &\leq -x_1^4-1.5x_1^2-\frac{1}{2}(x_1-x_2)^2, \label{lyapu}
\end{align}
thus \eqref{este} is GAS.
Additionally, we have that
\begin{equation}
\frac{\partial h(x)}{\partial x} \bigg|_{x=0}= \begin{bmatrix}
-2 & 1 \\ 
0 & -2
\end{bmatrix} \label{matrix1}
\end{equation}
is Hurwitz. By \cite[Corollary 4.3]{khalil2002nonlinear} \eqref{este} is LES, and consequently GALES. 

Next, we will derive a stabilizing sampling-period dependent DT control law. 
Let $H^e(x,T)$ denote the solution of \eqref{este} from initial state $x$, evaluated $T$ time units after. In other words, $H^e$ denotes the exact DT model of $\dot x = h(x)$. 
Additionally, according to Section~\ref{sec:preliminaries}, $F^e(x,u,T)$ denotes the solution of \eqref{continuo} from initial state $x$ and under a constant input $u$ (zero-order hold), evaluated $T$ time units after.
If a value of $u$ exists so that the matching equation $F^e(x,u,T)=H^e(x,T)$ is satisfied, then the solution of the sampled-data system would be equal to that of the desired closed-loop CT system \eqref{este} at a sampling instant.
For analytic and complete vector fields $\tilde f$ and $g$, the matching equation $F^e(x,u,T)=H^e(x,T)$
is solvable if and only if there exists a smooth function $\alpha: \R^n \rightarrow \R$ such that
$\frac{\partial g(x)}{\partial x}\tilde f(x)-\frac{\partial \tilde f(x)}{\partial x}g(x)=\alpha(x) g(x)$ 
\citep[Theorem 3.2]{monaco2007advanced}. 
This is not the case for the present plant since we have
\begin{equation}
\frac{\partial g(x)}{\partial x}\tilde f(x)-\frac{\partial \tilde f(x)}{\partial x}g(x)=\begin{bmatrix}
 -1 \\
 2x_2(x_2 - 1) + x_2^2 
\end{bmatrix} \neq \begin{bmatrix}
 0 \\
 \alpha(x) 
\end{bmatrix}. \label{esto4}
\end{equation}
Additionally,
no sampled\-/data control law can satisfy $|F^e(x,u,T)-H^e(x,T)|\leq C T^N$ for order $N\geq4$ for any $T>0$ sufficiently small and $C>0$ \citep{GRUNE2008225}. 

Since the matching equation is not solvable, we will propose
an approximate matching equation based on the Heun model of \eqref{este}
\begin{equation}
H^{Heun}(x,T):=x+\frac{T}{2}\lp h(x)+ h(x+T h(x)) \rp \label{videolla}
\end{equation}
to derive the control law.
If $H^e$ is approximated by $H^{Heun}$ 
and $F^e$ by $F^E$ we obtain
\begin{align}
 F^{E}(x,u,T)&= H^{Heun}(x,T) \es
 x+T[ \tilde f(x)+g(x)u]&=  H^{Heun}(x,T)\es
g(x)u&=\frac{ H^{Heun}(x,T)-x}{T} - \tilde f(x). \label{40} 
\end{align}

%
%
%

%
\sloppy Left-multiplying both sides of \eqref{40} by $g^\dagger(x):=(g^T(x)g(x))^{-1}g^T(x)$ and operating we have
       $U^{E/Heun}(x,T) :=g^\dagger(x) \left[\frac{ H^{Heun}(x,T)-x}{T} - \tilde f(x)\right ]$.
Replacing \eqref{videolla} into the last expression 
and solving yields
\begin{align}
&U^{E/Heun}(x,T)=- 2x_2\es
&+ \lp 1.5x_2^5 - 2.5x_2^4 + 5x_2^3 - 3x_2^2 + 2x_2\rp T \es
&+ \lp- 1.5x_2^7 + 3.5x_2^6 - 8.5x_2^5 + 8.5x_2^4 - 8x_2^3 + 2x_2^2\rp T^2  \es
&+\lp 0.5x_2^9 - 1.5x_2^8 + 4.5x_2^7 - 6.5x_2^6 + 9x_2^5 - 6x_2^4 + 4x_2^3 \rp T^3. 
\end{align}
We next simulate the sampled\-/data system defined by \eqref{continuo} 
for two different controllers: emulation given by
$U_c(x,T)=u_c(x)=-2x_2$ and
$U^{E/Heun}$.
Note that $U^{E/Heun}(x,0)=u_c(x)$, as expected.

Assumptions \ref{ass:B1} and \ref{ass:B2} 
are easy to verify.
Note that
$f(x,u)= \tilde f(x)+g(x)u$ 
fulfills Assumption \ref{ass:B1} if and only if  $\tilde f$ and $g$
are locally Lipschitz.
It is easy to prove that the pair $(U_c,U^{E/Heun})$ is StC
and that $U_c$ and $U^{E/Heun}$ are both StL. 
%
Since $\dot x=h(x)$ is GALES, Theorem \ref{TH6} establishes that both
$\bar F^e_{U_c}$ and $\bar F^e_{U^{E/Heun}}$ are SLES-VSR.

Figure~\ref{figure1} shows simulations from initial condition $x_0=[-1,1]^T$ for both sampled\-/data closed\-/loop models
for a constant sampling period $T=0.75$. 
%
Note that for the used sampling instant sequence, emulation leads to unstable behaviour while the sampled\-/data evolution corresponding to the controller $U^{E/Heun}$ is stable and gets closer to the exact continuous\-/time solution $H^e$. 

\section{Conclusions}
\label{sec:conclusions}
We presented mild consistency and regularity conditions on the plant and control laws
that allow to establish novel relationships between the stability of a sampled\-/data system fed back with sampling-period-dependent
control laws with the stability of the CT closed\-/loop plant obtained in the limit as the sampling period converges to zero. The given results extend previous ones under milder assumptions.

\section{Proofs} 
\label{sec:proofs}

{\color{black}
\subsection{ Proof of Theorem \ref{theorem:SEISS1}}
\label{section:proof:SEISS1}
\indent\emph{i).} 
We will use the results in \cite{IEEETAC18}. 
To do so
define $G^a(x,u,T):= \bar F^a(x,T)$ and $G^b(x,u,T):= \bar F^b(x,T)$ for any $u$. 
We will now prove that the fact that $(\bar F^a, \bar F^b)$ is EPC implies 
that $(U,G^a)$ is MSEC with $(U,G^b)$ for any control law $U(x,T)$ according to \citep[Definition 2.6]{IEEETAC18}.

%

Given $\X \subset  \R^n$ compact, define 
$M:=\max \{|x|:x \in \X\}$. Let the EPC 
definition generate $K>0$, $T^*>0$ and $\rho \in \K_\infty$ 
such that 
\begin{align}
&\left|\bar F^a (x,T)- \bar F^b(y,T) \right|    \leq (1+KT)\left|x-y\right|+ T \rho(T)  \max\{|x|,|y|\}
\end{align}
for all $|x|,|y| \leq M$ and $T\in (0,T^*)$.
Define $\rho_0 \in \K$ via $\rho_0 :=M\rho$ 
and the constant (hence nondecreasing) function $\sigma : \R_{\ge 0} \to \R_{\ge 0}$, $\sigma\equiv K$. 
Then, 
\begin{align*}
    |\bar F^a(x,T)-\bar F^b(y,T)| &= |G^a(x,u,T)-G^b(y,u,T)|\\
    &\leq T \rho_0(T) + (1 + T \sigma (T ))|x - y|
\end{align*}
holds for all $u$ and $x,y\in \X$. In particular, for any control law $U$, if we substitute $u = U(x+e,T)$, it follows that
\begin{align*}
    |\bar G^a_U(x,e,T) - \bar G^b_U(y,e,T)| &:=
    |G^a(x,U(x+e,T),T)-G^b(y,U(y+e,T),T)|\\
    &\leq T \rho_0(T) + (1 + T \sigma (T ))|x - y|,
\end{align*}
where we have employed the notation $\bar G^a_U(x,e,T)$ according to \citet{IEEETAC18}.
Therefore, the pair $(U,G^a)$ is MSEC with $(U,G^b)$ \citep[Definition 2.6]{IEEETAC18}.
%
%
%
By \citep[Theorem 3.1]{IEEETAC18} this last fact
is sufficient to ensure that 
if $x_{k+1}= \bar G^a_U(x_k,e_k,T_k)$ is Semiglobally Practically Input-to-State Stable under nonuniform sampling (SP-ISS-VSR) as defined in \citep[Definition 2.1]{IEEETAC18} then so is 
$x_{k+1}= \bar G^b_U(x_k,e_k,T_k)$ and viceversa. 
Note that the property SP-ISS-VSR \citep[Definition 2.1]{IEEETAC18} becomes SPS-VSR (Definition~\ref{def:sampleddef}) in the absence of errors (inputs).
Given that that $\bar G^a_U(x,e,T)= \bar F^a(x,T)$ and  $\bar G^b_U(x,e,T)= \bar F^b(x,T)$ for all $e \in \R^n$ the result follows.
\mer

\indent\emph{ii).} 
\label{app:proof:LES}
We will use the results in \cite{valcar_auto20}.
To do so define $\bar G^a(x,e,T):=\bar F^a(x,T)$ and $\bar G^b(x,e,T):= \bar F^b(x,T)$
for any $e$.
Note that if the pair $(\bar F^a,\bar F^b)$ is EPC then the pair
$(\bar G^a, \bar G^b)$ is REPC as defined in \citep[Definition 3.1]{valcar_auto20}. 
Furthermore, the REPC property implies the REPMC property \citep[Lemma 3.7]{valcar_auto20},
which is required to prove \citep[Theorem 3.5]{valcar_auto20}.


The current proof is based on slight modifications of the proof of 
\citep[Theorem 3.5]{valcar_auto20}. We modify its first part
to adapt it to the LES-VSR property under consideration:

Let $K_a\geq 1$ and $R_a,T^a,\lambda_a>0$,
characterize the 
LES-VSR 
property of 
$x^a_{k+1}=\bar G^a(x^a_k,e_k,T_k)$.
Thus, 
 for all
$k\in \N_0$, $\{T_i\}\in \Phi(T^a)$ and $|x_0|\leq R_a$ the solutions satisfy
$|x^a_k|\leq K^a |x^a_0|e^{-\lambda_a \sum_{i=0}^{k-1}T_i}\leq K^a R_a=:R$.
Let $\delta \in(0,1)$ and  $\eta\in(0,\delta)$.
Let $\T:=\frac{1}{\lambda_a} \ln\lp{\frac{K_a}{\delta-\eta}} \rp$.
Define $\T_1:=\T+1$ and
let \citet[Lemma 3.4]{valcar_auto20} generate
$T^L=T^L(R,0,\T_1,\eta)$.
Define $\bar T < \min \{1, T^a,T^L\}$.
Consider sampling period sequences such that $\{T_i\} \in \Phi(\bar T)$, and
for every $k\in\N_0$ and $j\in \N$ define
$s(k):= \sup \left\{r\in \N_0 : r \ge k+1, \sum_{i=k}^{r-1} T_i \leq \T_1 \right\}$ 
and
$s^j(k):=\overbrace{s(\hdots s(s}^j(k)))$.
From this expression the proof follows identically as in
the proof of \citep[Theorem 3.5]{valcar_auto20} by performing the following minor changes:
a) the quantity $E$ must be chosen as $E=0$, 
which implies that the error input sequence satisfies $e_i=0$ for all $i\in \N_0$,
b) rename $M$ and $M_a$ by $R_a$ and $R$, respectively.
Consequently, according to the proof of \citep[Theorem 3.5]{valcar_auto20}, we obtain that the state evolution for the model $\bar G^b$ from the initial state $\xi$ satisfies the following condition
\begin{align}
|x^b_k(\xi)| &\leq K_b
\exp{\left(- \lambda_b \sum_{i=0}^{k-1} T_i\right)} |\xi|
\end{align}
for all $k\in \N_0$, $|\xi|\leq R$ and $\{T_i\}\in \Phi(\bar T)$,
where $K_b:= (K_a+\eta)\exp{\left( \lambda_b \T_1\right)}$
$=(K_a+\eta)/\delta$ and $\lambda_b:=\ln({1/\delta})/\T_1$ and the result follows. \mer


\indent\emph{iii).} 
According to definition of SLES-VSR in Definition~\ref{def:EPC},
this result is a direct consequence of the results in the previous items \ref{item:SPS12}) and
\ref{item:LES12}). \mer

\qed

}

\subsection{Proof of Theorem \ref{suficientes1}.}
\label{app:Transitividad}
Consider $M \geq 0$ given and $|x|,|y|\leq M$. 
Let the StC property of $(U,V)$ generate $T^V:=T^V(M)>0$ and $\tilde \rho \in \K_\infty$ and the StL property of $U$ generate $K_U:=K_U(M)>0$ and $T^U:=T^U(M)>0$.
Thus $|U(x,T)|\leq K_UM$ for all $|x|\leq M$ and $T\in[0,T^U)$. We can bound $|V(x,T)|=|V(x,T)-U(x,T)+U(x,T)|\leq|V(x,T)-U(x,T)|+|U(x,T)|\leq \tilde \rho(T^V)M+K_UM=:E$.
Let the StLC property of $F^a$ generate $K(M,E)>0$ and $T^*:=T^*(M,E)>0$.
Define $\bar K:=K(1+K_{U})$
and $ \rho \in \K_\infty$ via $ \rho(s):=K\tilde \rho$.
Define $T^\star:=\min \{T^*,T^V,T^U\}$
Then, for all $|x|,|y|\leq M$ and $T\in[0,T^\star)$ we have
\begin{align}
&\left| F^a (x,U(x,T),T)- F^a (y,V(y,T),T) \right|   \es
&\leq (1+KT)|x-y|+KT|U(x,T)-V(y,T)| \label{1213} \\
&= (1+KT)|x-y|+KT|U(x,T)-U(y,T) +U(y,T)-V(y,T)| \es
&\leq (1+KT)|x-y|+KK_{U}T|x-y|+K T\tilde\rho(T)|y| \label{1212} \\
&\leq (1+\bar K T)|x-y|+ T   \rho(T)\max\{|x|,|y|\} \notag 
\end{align}
In \eqref{1213} and  \eqref{1212} we have used the facts that $F^a$ is StLC, 
$U$ is StL and $(U,V)$ is StC, respectively.
\qed   

\subsection{Proof of Theorem \ref{TH6}}
To prove Theorem~\ref{TH6} we need  Theorem~\ref{suficientes1} and 
the following Lemma~\ref{gales->SPSandLES} and Proposition \ref{TH1},
whose proofs are given at the end of this section.
Lemma \ref{gales->SPSandLES} establishes the relationship between CT 
stability properties of \eqref{lazocer2} and the DT model $H^e$, given
by its samples; i.e. $H^e(x_k,T_k)$ is just the solution of \eqref{lazocer2}
a time $T_k$ ahead from initial state $x_k$. 
{\color{black}
}

\begin{lemma} 
\label{gales->SPSandLES}
The CT closed\-/loop plant \eqref{lazocer2} is
\begin{inparaenum}[i)]
\itemsep0em
\item GAS $\Leftrightarrow$ $H^e$ is SPS-VSR. \label{item:1_g}
\item LES $\Leftrightarrow$ $H^e$ is LES-VSR. \label{item:2_g}
\item GALES $\Leftrightarrow$ $H^e$ is SLES-VSR. \label{item:3_g}
\item GES $\Leftrightarrow$ $H^e$ is SES-VSR. \label{item:0_g}
\end{inparaenum}
\end{lemma}
Proposition~\ref{TH1} shows that
under Assumptions \ref{ass:B1} and \ref{ass:B2}, 
the checkable mild sufficient conditions of
StL for the control law $U$ 
and StC for the pair $(U,U_c)$
ensure that the pair $(H^e, \bar F^{e}_U)$ is EPC. This establishes a correspondence between CT and sampled\-/data systems via Lemma \ref{gales->SPSandLES}.

\begin{prop} 
\label{TH1}
Under the assumptions of Theorem \ref{TH6},
 the pair 
$(H^e, \bar F^{e}_U)$ is EPC.
\end{prop}

{\color{black}
Note that asymptotic stability of \eqref{lazocer2} does not imply that $\bar F^e_U$ also exhibits asymptotic 
properties. Under the assumptions of Theorem \ref{TH6}, by Lemma~\ref{gales->SPSandLES}, Proposition~\ref{TH1} and Theorem~\ref{theorem:SEISS1},
the system \eqref{lazocer2} is GAS if and only if $\bar F^e_U$ is SPS-VSR (however GAS of \eqref{lazocer2} does not imply SS-VSR of $\bar F^e_U$ as one might intuitively think).
}
Next we use Theorem~\ref{suficientes1}, Lemma~\ref{gales->SPSandLES} and the following Lemma~\ref{lem:impli} to prove Proposition~\ref{TH1}. Recall that $U_c(x,T):=u_c(x)$ for all $T$.

\begin{lemma}
\label{lem:impli}
The following implications hold
\begin{enumerate}[i)]
\itemsep0em
    \item Assumptions \ref{ass:B1} and \ref{ass:B2}   $\Rightarrow$ $h$ is locally Lipschitz.  
     \label{clm1}
        \item \label{claim_1}   $h$ is locally Lipschitz  $\Rightarrow$
        $H^e$ is StLC. Moreover, for each $M \geq 0$ there exists $\bar T:=\bar T(M)$ such that for all $|x|\leq M$ and $T\in(0, \bar T)$ we have $|H^e(x,T)|\leq 2M$.
    \item
    $h$ is locally Lipschitz  $\Rightarrow$
$(H^e, \bar F^{E}_{U_c})$ is EPC.
        \label{lem:REPC0}
     \item Assumption \ref{ass:B1} 
     $\Rightarrow$ $F^{E}$ (Euler model) is StLC.    \label{lem:A2}
     \item Assumption \ref{ass:B2} 
     $\Rightarrow$ $U_c$ is StL.  \label{lem:A3}
    \item   $F^{E}$ is StLC+ $U$ is StL+ $(U_c,U)$ is StC   
 $\Rightarrow$ $(\bar F_{U_c}^{E}, \bar F^{E}_U)$ is EPC.
\label{lem:REPC1}
\item
\label{lem:REPC2}
 Assumptions \ref{ass:B1} and \ref{ass:B2} +$U$ is StL 
$\Rightarrow$
$( \bar F^{E}_U,\bar F^e_U)$ is EPC.
\end{enumerate}
\end{lemma}

\begin{proof} 
i) Consider $M\geq 0$ given.
Define $L_{u_c}:=L_{u_c}(M)$ from Assumption \ref{ass:B2}, $L_f:=L_f(M,L_{u_c})$ from Assumption \ref{ass:B1} and $L:= L_f(1+L_{u_c})$.
For all $|x|,|y| \leq M$ we have
\begin{align*}
|h(x)-h(y)|&= |f(x,u_c(x))-f(y,u_c(y))| \leq L_f(|x-y|+|u_c(x)-u_c(y)|) 
 \\ &\leq L_f(1+L_{u_c})|x-y|=L|x-y|. 
\end{align*}
\mer

ii) Consider $M\geq 0$ given.
Let the locally Lipschitz property of $h$
generate 
$L_h:=L_h(M)>0$. Define $\bar L (M):=L_h(2 M)>0$
and $\bar T(M):= \ln{(2)}/\bar L$.
We claim that $| H^e(x,T)| \leq 2M$ for all $|x|\leq M$ and $T \in (0, \bar T)$. For a contradiction, let $x$ and $T'$ be such that $|x|\le M$, $T' \in (0,\bar T)$ and $|H^e(x,T')| > 2M$. Define $\tau := \inf \{T > 0 : |H^e(x,T)| > 2M\}$. Then $|H^e(x,T)| \le 2M$ for all $T \in (0,\tau]$ and $|H^e(x,\tau)| = 2M$ by continuity, with $\tau < \bar T$. We have
    $|H^e(x,\tau)| \le |x| + \int_0^\tau |h(H^e(x,s))| ds
    \le M + \bar L \int_0^\tau |H^e(x,s)| ds.$
Using Gronwall inequality, then $|H^e(x,\tau)| \le M e^{\bar L \tau} < M e^{\bar L \bar T} = 2M$, reaching a contradiction. The claim is thus true.
For all $|x|,|y| \leq M$ and $T \in (0, \bar T)$ then
\begin{align*}
 \lb H^e(x,T)  - H^e(y,T)\rb &\leq \lb x-y \rb  +  \int_0^T  \lb h(H^e(x,s) ) -h(H^e(y,s))  \rb  ds \es
&\leq \lb x-y \rb  + \bar L \int_0^T  \lb H^e(x,s)  -H^e(y,s)  \rb  ds.
\end{align*}
By Gronwall inequality then 
$| H^e(x,T)  - H^e(y,T)|  \le e^{\bar L \bar T} \lb x-y \rb$,  
thus $H^e$ is StLC. \mer

iii) Consider $M>0$ given. Let the StLC property of $H^e$ generate $L_H:=L_H(2M)$ and $\bar T:=\bar T(M)$. From the claim in item~ii), note that  $| H^e(x,T)| \leq 2M$ for all $|x|\leq M$ and $T\in (0,\bar T)$.
Define $L_h:=L_h(2M)$ from the fact that $h$ is locally Lipschitz.
Thus, for all $|x|,|y| \leq M$ and $T\in (0, \bar T)$ we have 
\begin{align}
\allowdisplaybreaks
&| H^e(x,T)-\bar F^{E}_{U_c}(y,T)| =\left|x+\int_{0}^{T} h( H^e(x,s)) ds - y-Th(y) \right|   \es
&\leq \left|x - y \right |+L_h \int_{0}^{T}  \left | H^e(x,s) -y \right| ds     \es
&\leq \left|x - y \right |+L_h \int_{0}^{T}  \bigg | H^e(x,s) -y -\bar F^{E}_{U_c}(y,s) +\bar F^{E}_{U_c}(y,s) \bigg| ds     \es
&\leq \left|x - y \right |+L_h \int_{0}^{T}  \left | H^e(x,s)-\bar F^{E}_{U_c}(y,s) \right|ds  +L_h \int_{0}^{T} T \left| h(y) \right| ds     \es
&\leq \left|x - y \right |+L_h \int_{0}^{T}  \left | H^e(x,s)-\bar F^{E}_{U_c}(y,s) \right|  ds  +L_h T^2\left| h(y) \right|. \notag  
\end{align}
By Gronwall inequality,
\begin{align*}
\left| H^e(x,T)-\bar F^{E}_{U_c}(y,s) \right| &\leq  \left(\left|x - y \right |+L_h T^2\left| h(y) \right| \right) e^{L_h T} 
\es &
\le |x-y| \left( 1 + \frac{e^{L_h T} - 1}{T}\,T \right) + L_h^2 T^2 e^{L_h T}|y|
\es &
\leq   \left(1+KT \right) \left|x - y \right |  +\rho(T) T \left| y \right|
\end{align*}
where $K:=(e^{L_h \bar T} - 1)/\bar T = \left(L_h+\sum_{k=2}^\infty \frac{L_h^k \bar T^{k-1}}{k!}  \right)$
and $\rho \in \K_\infty$ is defined as $\displaystyle\rho(s):= L_h^2 e^{L_h \bar T} s$. 
\mer 

iv) Consider $M,D \geq 0$ given, then
$|F^{E}(x,u,T)-F^{E}(y,v,T)|  \leq |x-y|+ T|f(x,u)-f(y,v)| 
\leq |x-y|+ T L(|x-y|+|u-v|) \leq (1+KT)|x-y|+KT|u-v|$
for all $|x|\leq M$, $|u|,|v|\leq D$, and $T\in[0,\infty)$
where we have used Assumption \ref{ass:B1}  and  $K:=L$. 
\mer 

v) 
By Assumption \ref{ass:B2}  $u_c(x)$ is locally Lipschitz. 
Given that it does not depend on $T$ the result is immediate. 
\mer 

vi) Conditions of  Theorem \ref{suficientes1} hold, thus the result is immediate. 
\mer 

vii) 
We will prove that Assumptions 2.1-2.3 (A2.1-2.3 in the following) and conditions i) and ii) of \cite[Theorem~3.9]{valcar_auto20} hold. 
Define 
$\bar U(x,e,T):=U(x,T)$ for each $q\in \N$ and all $e \in \R^q$.
Thus, the closed-loop Euler and exact models result $\bar F^E_U(x,e,T):=\bar F^E(x,U(x,e,T),T)$ and
 $\bar F^e_U(x,e,T):=\bar F^e(x,U(x,e,T),T)$, respectively.

A2.1: 
It is a
direct consequence of Assumption \ref{ass:B1}. 

A2.2: 
Consider $M,C_u \geq 0$ given and $|x|\leq M$ and $|u|\leq C_u$.
From Assumption \ref{ass:B1} define $L:=L(M,C_u)$ and
function $C_f(r,s):=L(r+s)$, then for all $|x|\leq M$ and $|u|\leq C_u$
$|f(x,u)|=|f(x,u)-f(0,0)|\leq L(|x|+|u|)  
\leq L(M+C_u)
:=C_f(M,C_u)$. 

A2.3:
Consider
$M,E\geq 0$ given. 
Define $L:=L(M)$ from Assumption \ref{ass:B2}
and $K:=K(M)$, $T^*(M)$ from the fact that $U$ is StL.
Define $T_u(M,E):=\min\{1, T^*(M)\}$ and
for all $|x|\leq M$, $|e|\leq E$ and $T \in [0, T_u)$, we have 
$|\bar U(x,e,T)|\leq |U(x,T)-U(0,T)|+|U(0,T)| \leq K|x| 
\leq K M =:C_u(M,E)$.  

Condition i): 
The result is immediate, since we have that $f(0,\bar U(0,e,T))=f(0,U(0,T))=f(0,0)=0$.

Condition ii): 
From  A2.3,
and $T\in(0,T_u)$ we have that $|\bar U(x,e,T)|\leq C_u(M,E)$.
Define
$L:=L(M,C_u(M,E))$ from Assumption \ref{ass:B1}.
For all $ |x^a|,|x^b|\leq M$, $|e|\leq E$ and $T\in(0,T_u)$ we have
$|f(x^a,\bar U(x^a,e,T))-f(x^b,\bar U(x^b,e,T))| 
\leq L(|x^a-x^b|+|\bar U(x^a,e,T)-\bar U(x^b,e,T)|) 
\leq L(|x^a-x^b|+K|x^a-x^b|) \leq L(1+K) |x^a-x^b|$.
Thus, by \cite[Theorem~3.9]{valcar_auto20} the pair $( \bar F^{RK}_U,\bar F^e_U)$ is REPC. By assuming that no errors affect the control input, i.e. $e=0$, REPC coincide with the EPC property and therefore $(\bar F^{RK}_U,\bar F^e_U)$ is EPC for any explicit Runge-Kutta (RK) model.
Given that the Euler model is the simplest RK model,  $(\bar F^{E}_U,\bar F^e_U)$ is EPC and the result follows. 
\mer 
\end{proof}

\begin{proof}[Proof of Lemma \ref{gales->SPSandLES}] 

\indent\emph{Proof of item \ref{item:1_g})} 

$\Longrightarrow$)
Let the GAS property of \eqref{lazocer2}
generate $\beta \in \KL$ from \ref{item:gas}) of Definition \ref{def:conti}.  
We have
$
|x(t)| \leq  \beta(|x_0|, t), \medspace \forall t \geq 0$.
Then,
\begin{equation}
    |x_k| \leq  \beta\lp |x_0|,\sum_{i=0}^{k-1} T_k \rp  \label{eq:SS_2}
\end{equation}
for all $k \in \N_0$, $\{T_i\} \in \Phi(\infty)$ and $x_0\in \R^n$. 
Thus,
$H^e$ is SS-VSR with $T^*:=\infty$, and hence SPS-VSR by the fact that SS-VSR implies SPS-VSR.

$\Longleftarrow$)
Let the SPS-VSR property of $H^e$ generate
$\beta \in \KL$ and $T^*:=T^*(M,R)$ for every $M,R>0$
such that the bound 
$|x_k| \leq  \beta\lp |x_0|,\sum_{i=0}^{k-1} T_k \rp+R$
holds for all $k \in \N_0$, $\{T_i\} \in \Phi(T^*)$ and $|x_0|\leq M$. We next establish GAS with function $\bar \beta := 2\beta \in \KL$.

Consider $x_0 \in \R^n$ and $t\geq 0$ given.  
Define $M:=|x_0|$ 
and the constants $R\in(0,1)$ such that $R \leq \beta (M,t)$ and
$N\in \N$ such that $N\geq \frac{2t}{T^*(M,R)}$.
Define the
constant sampling period sequence as $T_k\equiv \frac{t}{N} \le T^*/2$.
Note that
$\sum_{i=0}^{N-1} T_i = t$ and $x(t)=x_N$. Thus,
\begin{equation}
    |x(t)| \leq \beta\lp |x_0|,\sum_{i=0}^{N-1} T_i \rp+R 
    =  \beta\lp |x_0|,t \rp+R \leq
    \bar \beta \lp |x_0|,t\rp. \label{eqcota1235}
\end{equation}
Given that \eqref{eqcota1235} holds for any 
given $x_0 \in \R^n$ and $t\geq 0$ the result follows. \mer

\indent\emph{Proof of item \ref{item:2_g})} 

$\Longrightarrow$)
Let the LES property of \eqref{lazocer2}
generate  $K\geq 1$ and $R,\lambda>0$.   
We have $|x(t)|\leq K |x_0|e^{-\lambda t}, \medspace \forall t\geq 0$.
Then, $|x_k|\leq K |x_0|e^{-\lambda \sum_{i=0}^{k-1}T_i}$
for all
$k\in \N_0$, $\{T_i\}\in \Phi(\infty)$ and $|x_0|\leq R$ and thus, $H^e$ is LES-VSR.

$\Longleftarrow$)
Let the LES-VSR property of $H^e$ generate
$K\geq 1$ and $R,T^*,\lambda>0$ and consider $|x_0|\leq R$ given
such that the bound 
$|x_k|\leq K |x_0|e^{-\lambda \sum_{i=0}^{k-1}T_i}$
holds for all $k \in \N_0$, $\{T_i\} \in \Phi(T^*)$ and $|x_0|\leq R$. 
Consider a sampling period sequence such that $T_k\equiv \frac{aT^*}{2}$ with $a\in(0,1]$ for all $k \in \N_0$.
Define $t(k):=\frac{aT^*}{2}k$. Thus, for the given initial condition $x_0$ we have
$|x_k| \leq  K |x_0| \exp \lp-\lambda t(k) \rp$ 
for all $k\in \N_0$.
Given that, irrespectibly of $x_0 \in \R^n$, the  previous bound
holds for any $a\in(0,1]$ it also holds for any $t\in \R_{\geq 0}$ and the result follows. \mer

\indent\emph{Proof of item \ref{item:3_g})} 
The result follows directly from the proofs of items \ref{item:1_g}) and \ref{item:2_g}) and the 
GALES definition. 
\mer

\indent\emph{Proof of item \ref{item:0_g})} 
$\Longrightarrow$)
Let the GES property of \eqref{lazocer2}
generate $\lambda >0$ and $K\geq 1$.
We have
$
|x(t)| \leq  K |x_0| \exp(-\lambda t), \medspace \forall t \geq 0$.
Then 
\begin{equation}
    |x_k| \leq  K |x_0| \exp \lp-\lambda \sum_{i=0}^{k-1} T_k \rp  \label{eq:SES_2}
\end{equation}
for all $k \in \N_0$, $\{T_i\} \in \Phi(\infty)$ and $x_0\in \R^n$. 
Thus,
$H^e$ is SES-VSR with $T^*:=\infty$.

$\Longleftarrow$)
Consider that $H^e$ is SES-VSR. 
Consider $x_0 \in \R^n$ given. Define $M:=|x_0|$, and
let the SES-VSR property of $H^e$ generate $\lambda>0$, $K\geq 1$ and $T^*:=T^*(M)$ such that the bound \eqref{eq:SES_2} holds for all $k \in \N_0$, $\{T_i\} \in \Phi(T^*)$ and the given $x_0$. 
Consider a sampling period sequence such that $T_k\equiv \frac{aT^*}{2}$ with $a\in(0,1]$ for all $k \in \N_0$.
Define $t(k):=\frac{aT^*}{2}k$. Thus, for the given initial condition $x_0$ we have
$    |x_k| \leq  K |x_0| \exp \lp-\lambda t(k) \rp $
for all $k\in \N_0$.
Given that, irrespectibly of $x_0 \in \R^n$, this last bound holds for any $a\in(0,1]$ it also holds for any $t\in \R_{\geq 0}$ and the result follows.
\mer

\end{proof}

\begin{proof}[Proof of Proposition \ref{TH1}]
\label{app:TH1}
By the implications in Lemma \ref{lem:impli}, the pairs $( H^e, \bar F^{E}_{U_c})$, $( \bar F^{E}_{U_c}, \bar F^{E}_{U})$ and $(\bar F^{E}_U,\bar F^e_U)$ are EPC. 
By the transitivity of the EPC property  $( H^e,\bar F^e_U)$ is EPC.
\end{proof}

\begin{proof}[Proof of Theorem \ref{TH6}]
\label{app:theorem7}
Conditions of Proposition~\ref{TH1} imply that $(H^e,\bar F^e_U)$
is EPC. Consider that $\bar F^e_U$ or $H^e$ fulfills the SES-VSR, LES-VSR or SLES-VSR properties.
By Theorem \ref{theorem:SEISS1}, $H^e$ or $\bar F^e_U$ fulfills the same properties, respectively. By the implications of Lemma \ref{gales->SPSandLES} the result follows. \end{proof}


\bibliographystyle{elsarticle-harv}
\bibliography{bibliografia201810}

\begin{figure}[ht]
\includegraphics[width=\textwidth]{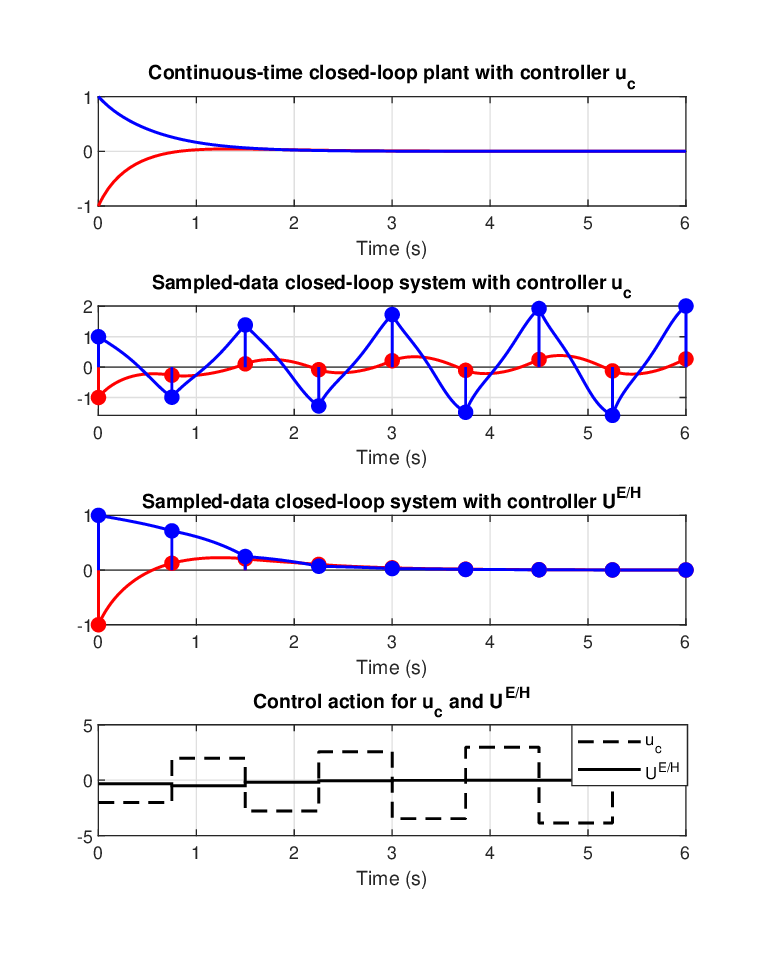}
\vspace{-1cm}
\caption{Evolution of the continuous\-/time closed\-/loop plant with $u_c=-2x_2$ compared with the
evolution of
the sampled\-/data  system \eqref{continuo} in closed\-/loop with controllers $u_c$ and $U^{E/Heun}$ for a constant sampling period $T=0.75$. Blue: $x_1$. Red: $x_2$.}
\label{figure1}
\end{figure}

\end{document}